\documentclass[12pt,a4]{article}
\usepackage[T1]{fontenc}
\usepackage[margin=2.5cm]{geometry}
\usepackage{microtype}

\usepackage{amsmath,amssymb,amsthm,mathtools}
\usepackage{setspace,enumitem,soul,url,xcolor,hyperref,bm,pifont,cancel,comment,booktabs}
\usepackage{subcaption}
\usepackage{tikz-cd}

\newtheorem{lemma}{Lemma}
\newtheorem{coro}{Corollary}

\usepackage{authblk}
\author[1,2]{Li-Chun Zhang}
\author[3]{Tiziana Tuoto}
\affil[1]{\em \small University of Southampton, UK}
\affil[2]{\em \small Statistisk sentralbyrå, Norway}
\affil[3]{\em \small Istituto nazionale di statística, Italy}
\title{Conformal confidence intervals with an application to small area estimation}
\date{}

\begin{document}

\maketitle

\begin{abstract} Conformal prediction inference yields intervals for out-of-sample random outcomes with designated coverage probabilities, given exchangeable or independent-and-identically distributed random variables. For regression analysis, valid coverage can be achieved given finite sample sizes, despite unavoidable misspecifications of the regression function. We propose a novel method of conformal inference, aimed to produce confidence intervals of the unknown expectations of the in-sample outcomes with the designated coverage probabilities conditional on the realised sample. These conformal confidence intervals fill a gap between classical regression and conformal inference. The proposed approach is applied to small area estimation problems. 
\end{abstract}


\section{Introduction} \label{sec:introduction}

Regression is a cornerstone of statistical modelling to study or predict the outcome (or dependent) variable in relation to the available features (or covariates, explanatory variables), where the expectation of the outcome is formulated as a \emph{regression function} of the features. In the finite-sample setting, classical regression inference relies heavily on correct model specifications. For instance, for the pivotal student-$t$ distribution to provide valid predictive inference, one needs to have the correctly specified linear regression function, and the regression errors must be independent, homoscedastic, and normally distributed conditional on the features. Conformal methods dating back to Vovk et al. (2005) represent an alternative that has received attention in the recent years, because they can provide valid predictive inference even without correctly specifying the regression function or the error distribution. In particular, Lei et al. (2018) popularised the computationally effective \emph{split-conformal} prediction intervals and established the conditions for their asymptotic efficiency. 

Specifically, suppose regression data are observed for a sample of units $s = \{ 1, ..., m \}$, comprised of continuous outcomes $y_i \in \mathbb{R}$ and vectors of features $x_i\in \mathbb{R}^p$ from the joint distribution $\mathcal{P}$, denoted by
\[
\{ (y_i, x_i) : i\in s\} \sim \mathcal{P} ~.
\]
For an \emph{out-of-sample} outcome $y_{m+1}$ with the associated $x_{m+1}$, one only needs $(y_i, x_i)$ to be independent and identically distributed (IID), $i=1, ..., m+1$, for the conformal prediction interval to achieve the designated coverage probability of $y_{m+1}$ in the finite-sample setting, i.e. over the $m+1$ IID draws, even though the adopted regression function for $\mathbb{E}(y\mid x)$, denoted by $\mu(x)$, is always misspecified to a greater or lesser extent in practice. 

We consider a different regression inference problem that is often of interest. Let 
\[
\theta_i = \theta(x_i)
\]
be the unknown expectation of $y_i$ given $x_i$ with respect to $\mathcal{P}$. For each \emph{in-sample} unit, $i\in s$, we propose to construct $\mathcal{C}_i$ as an interval estimator of $\theta_i$, by some suitable conformal methods, which is aimed to achieve the chosen \emph{marginal confidence coverage} level $\alpha$ over the finite sample $s$, given as
\begin{equation} \label{cvr_conf}
\frac{1}{m} \sum_{i=1}^m \Pr\big( \theta_i \in \mathcal{C}_i \big) = \alpha ~.
\end{equation}

As far as we are aware of, such in-sample \emph{conformal confidence intervals} will fill a gap between classical regression and conformal inference, which can have wide-ranging applications wherever regression analysis is of interest.

\subsection{Small area estimation as a motivating application}

Denote by $\mathcal{D} =\{ 1, ..., m\}$ the areas (or domains), and let $Y_i$ be the area population parameter of interest for each $i\in \mathcal{D}$. Let $y_i$ be a \emph{direct} estimator of $Y_i$, which is based on the within-area sample that is of the size $n_i$. In small area estimation (SAE),  it is common to assume that each $y_i$ is unbiased of $Y_i$ over repeated sampling from the given population; see e.g. Fay and Herriot (1979), Rao and Molina (2015). Let the unknown $Y_i$'s be the parameters $\theta_i$'s in the general formulation above, for which we would like to obtain conference intervals with the coverage level \eqref{cvr_conf}.

The SAE challenge arises due to the limited area sample sizes $n_i$, such that the direct estimator $y_i$ has too large a variance to be acceptable in many (or most) areas. It is then common to adopt model-based methods in order to reduce the variance of estimation. For example, the FH-model is given by
\begin{equation} \label{FHmod}
y_i = x_i^{\top} \beta + v_i + e_i
\end{equation}
(Fay and Herriot, 1979), where $x_i$ is the vector of features and $\beta$ that of the regression coefficients, $v_i$ is the random effect and $e_i$ the sampling error of $y_i$. Both $v_i$ and $e_i$ have expectation zero, and they are assumed to be independent of each other. 

Under \eqref{FHmod}, one can obtain the empirical best linear unbiased predictor (EBLUP) of $\theta_i$, as well as an interval estimator of $\theta_i$ based on the estimated mean squared error (MSE) of the EBLUP. However, these EBLUP intervals may fail to achieve the nominal level of coverage due to several reasons. 
\begin{itemize}[leftmargin=6mm,itemsep=0pt] 
\item The adopted regression function $\mu(x_i) = x_i^{\top} \beta$ is misspecified.
\item The sampling variances, denoted by $\psi_i = \mathbb{V}(e_i)$, need to be estimated. 
\item The random effects $v_i$ may be heteroscedastic, i.e. $\mathbb{V}(v_i)$ is not a constant $\sigma_v^2$.
\end{itemize}
Notice that there are many other models used in SAE, such as the unit-level nested error regression model (Battese et al. 1988), where similar challenges exist for the corresponding model-based interval estimation. 

In any case, we shall let $\mu(x_i)$ denote generically any regression function targeted at the area-specific parameter $\theta_i$, such that the conformal methods we consider will be agnostic to the underlying model specifications.

\subsection{Conformal confidence intervals} 

Instead of interval estimators under models like \eqref{FHmod}, we shall consider conformal methods aimed at valid coverage \eqref{cvr_conf}, regardless the misspecifications of the adopted regression function $\mu(x_i)$ or the estimated variance of $y_i$, while it is still possible to utilise the relevant features in terms of $\mu(x_i)$ in order to reduce the interval widths. 

However, the existing conformal inference methods are concerned with \emph{prediction intervals} of out-of-sample random variables; see e.g. Vovk et al. (2005), Shafer and Vovk (2008), Vovk et al. (2022), Lei et al. (2018), Tibshirani et al. (2019). In the context of SAE, Bersson and Hoff (2024) consider conformal prediction inference of the unobserved population units from this perspective, where the target is the out-of-sample $y$-values rather than the area population parameters $Y_i$.

We therefore need to overcome two obstacles. First, we would like to establish the confidence coverage property \eqref{cvr_conf}, where each $\theta_i$ is a constant of sampling, i.e. we need conformal confidence intervals instead of prediction intervals. 

Second, insofar as the IID assumption of $(y_i, x_i)$, where $i\in s$, may be unnecessary or inappropriate, such as in the case of finite population sampling, we would like to achieve the confidence coverage \eqref{cvr_conf} conditional on the actual sample $s$, without resorting to the assumption of IID $(y_i, x_i)$ or exchangeable joint distribution of $\{ (y_i, x_i) : i \in s \}$.

\subsection{Outline of development} 

We shall develop conformal confidence intervals in two steps. Firstly, we introduce a method of conformal intervals for confidence inference, instead of prediction inference, under the assumption of IID $(y_i, x_i)$. Next, to apply the new conformal method, we propose a \emph{split-sampling} design conditional on $s$, such that the probability in \eqref{cvr_conf} can be evaluated with respect to the known split-sampling design, without the assumption of IID or exchangeability for the underlying joint distribution of $\{ (y_i, x_i) : i \in s\}$. 

These \emph{design-based conformal confidence intervals} will be developed in Section \ref{sec:CCI}, and applied to SAE in Section \ref{sec:application}. Some final remarks will be given in Section \ref{sec:final}.

\section{Design-based conformal confidence intervals} \label{sec:CCI}

\subsection{Conformal shrinkage interval} \label{sec:CSI}

Let us start with some simple probability statements. Given any IID sample $z_1, ..., z_n, z_j$, $j\notin \{ 1, ..., n\}$, where the set of units $\{ 1, ..., n\}$ will be referred to as the \emph{calibration set} for the \emph{holdout} unit $j$, we have
\[
\Pr\big( \min(z_1, ..., z_n) \leq z_j \leq \max(z_1, ..., z_n) \big) = \frac{n-1}{n+1} ~.
\]
Notice that the conformal interval above is asymmetric generally, in contrast to the more commonly used symmetric conformal interval 
\[
\Pr\big( -\max(|z_1|, ..., |z_n|) \leq z_j \leq \max(|z_1|, ..., |z_n|) \big) = \frac{n}{n+1} ~.
\]
We consider both types of conformal intervals since, as we shall see, their confidence coverage and relative efficiency depend on different conditions. 

Next, given IID $(y_i, x_i)$, and any known (or pre-trained) regression function $\mu_i =\mu(x_i)$, we have, by virtue of IID $z_i = y_i - \mu_i$,  
\begin{gather*}
\Pr\big( \mu_j + \min_{i=1,...,n} (y_i-\mu_i) \leq y_j \leq \mu_j + \max_{i=1,...,n} (y_i-\mu_i) \big) = \frac{n-1}{n+1}~, \\
\Pr\big( \mu_j - \max_{i=1,...,n} |y_i-\mu_i| \leq y_j \leq \mu_j + \max_{i=1,...,n} |y_i-\mu_i| \big) = \frac{n}{n+1}~.
\end{gather*}
In particular, we shall refer to 
\[
u_i = y_i - \mu_i
\] 
as the \emph{(conformal) score}, to be distinguished to the \emph{absolute scores} $|u_i| = |y_i - \mu_i|$. 

Now, for confidence inference, we propose to consider a new conformal method. Let $\phi$ be a shrinkage constant, $0< \phi <1$. Let the \emph{shrinkage score} be
\[
b_i(\phi) = y_i - \dot{\theta}_i(\phi) = (1-\phi) (y_i - \mu_i) \qquad\text{where}\qquad \dot{\theta}_i(\phi) = \phi y_i + (1-\phi) \mu_i ~.
\]
Let the \emph{asymmetric conformal shrinkage interval (ACSI)} for the holdout unit be 
\begin{equation} \label{eq:acsi}
\mathcal{A}_j(k; \phi) = [\dot{\theta}_j(\phi) + b_{(k)}(\phi),~ \dot{\theta}_j(\phi) + b_{(n -k +1)}(\phi)] 
\end{equation}
where $b_{(1)}, ..., b_{(n)}$ are the order statistics of $\{ b_i : i=1,..., n\}$ in the calibration set, and $k=1, ..., \lfloor \tfrac{n}{2}\rfloor$. Moreover, let the \emph{conformal shrinkage interval (CSI)} be
\begin{equation} \label{eq:csi}
\mathcal{C}_j(k; \phi) = \big[ \dot{\theta}_j(\phi) - |b|_{(n-k+1)}(\phi),~ \dot{\theta}_j(\phi) + |b|_{(n -k +1)}(\phi) \big] 
\end{equation}
where $|b|_{(1)}, ..., |b|_{(n)}$ are the order statistics of $\{ |b_i| : i =1, ..., n\}$, and $k=1, ..., n$.

\begin{lemma} \label{lemma:IID:pred}
Given any $\phi \in (0,1)$, and $\mu(x)$ that does not depend on the IID $(y_i, x_i)$, the prediction coverage of $y_j$ is $\alpha = 1- \frac{2k}{n+1}$ by $\mathcal{A}_j(k; \phi)$ and $\alpha = 1- \frac{k}{n+1}$ by $\mathcal{C}_j(k; \phi)$. 
\end{lemma}

The proof is given in Appendix \ref{sec:proof}, as are all the other proofs later. The reason is straightforward now that the shrinkage scores are IID, just like  the `full' absolute scores $|u_i| = |y_i -\mu_i|$ used in the standard conformal methods.

\begin{lemma} \label{lemma:IID:ACSI}
Given $\mu(x)$ that does not depend on the IID $(y_i, x_i)$, the confidence coverage of $\theta_j$ by the ACSI $\mathcal{A}_j(k; 0.5)$ with $\phi =0.5$ is $\alpha = 1- \frac{2k}{n+1}$ exactly, provided $e_i = y_i - \theta_i$ has a symmetric distribution, otherwise the coverage differs from $\alpha = 1- \frac{2k}{n+1}$ by
\begin{equation} \label{mis:ACSI}
\epsilon_{\mathcal{A}} = \Pr\big\{ u_{(k)} \leq 2(\theta_j - \mu_j) - u_j \leq u_{(n-k+1)} \big\} - \Pr\big\{ u_{(k)} \leq u_j \leq u_{(n-k+1)} \big\} ~.
\end{equation} 
\end{lemma}

\begin{lemma} \label{lemma:IID:CSI}
Given $\mu(x)$ that does not depend on the IID $(y_i, x_i)$, the confidence coverage of $\theta_j$ by the CSI $\mathcal{C}_j(k; 0.5)$ with $\phi =0.5$ is $\alpha = 1- \frac{k}{n+1}$ exactly (i) if $\mu(x_j) = \theta_j$ or (ii) if $e_i = y_i - \theta_i$ has a symmetric distribution, otherwise the coverage differs to $\alpha = 1- \frac{k}{n+1}$ by
\begin{equation} \label{mis:CSI}
\begin{split}
\epsilon_{\mathcal{C}} & = \Pr\big\{ -|u|_{(n-k+1)} \leq u_j - 2 (\theta_j - \mu_j) \leq |u|_{(n-k+1)} \big\} \\
& \qquad - \Pr\big\{ -|u|_{(n-k+1)} \leq u_j \leq |u|_{(n-k+1)} \big\} ~.
\end{split}
\end{equation} 
\end{lemma}

We notice that $\phi =0.5$ is the constant choice for confidence inference, either by ASCI or CSI, which does not require tuning in different applications. 
Notice that Lemmas \ref{lemma:IID:ACSI} and \ref{lemma:IID:CSI} reveal different conditions for exact confidence coverage: while both would achieve it given symmetric distribution of $y_i$ around its expectation $\theta(x_i)$, the CSI would also achieve exact coverage if $\mu_j = \theta_j$, whereas correct specification of $\mu_j = \theta_j$ by itself cannot guarantee exact coverage by the ACSI.

\subsection{Miscoverage (I)} \label{sec:mis_cvr_1}

When it comes to the miscoverage probability \eqref{mis:ACSI} by the ACSI and \eqref{mis:CSI} by the CSI, there exists an important difference. The CSI can approximately achieve the nominal level of coverage as long as a location shift from $u_j$ to $u_j -2(\theta_j - \mu_j)$ does not materially affect their coverage by $[-|u|_{(n-k+1)}, |u|_{(n-k+1)}]$; whereas the ACSI can do so if the same holds when $u_j -2(\theta_j - \mu_j)$ is flipped around to $2(\theta_j - \mu_j) - u_j$ in addition, which only depends on the symmetry of $y_j - \theta_j$ but not $\mu_j -\theta_j$, since the expectation of $u_j$ is equal to the expectation of $2(\theta_j - \mu_j) - u_j$. 

\begin{figure}[ht]
\centering
\includegraphics[scale=0.55]{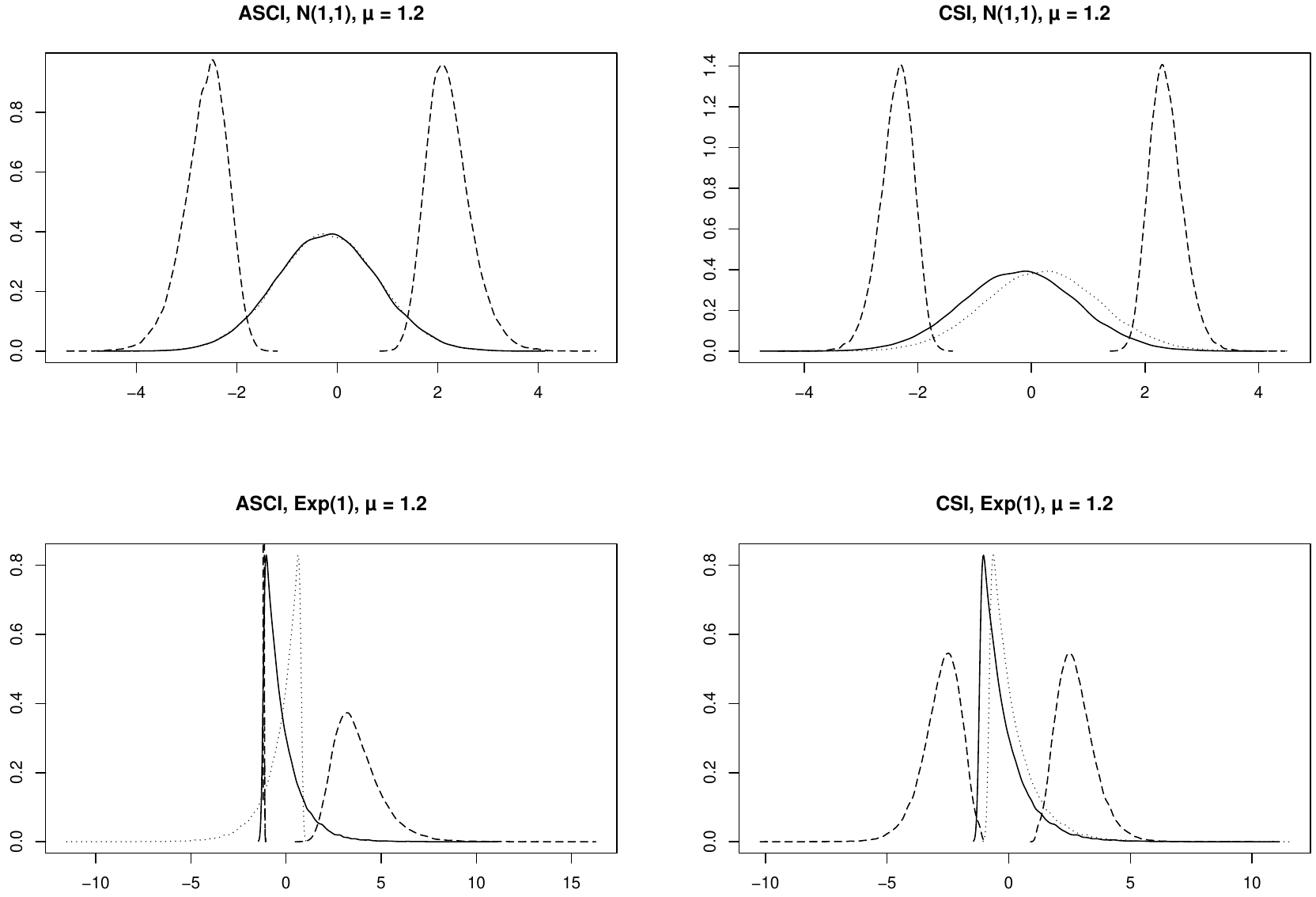}
\caption{Density of $u_j$ (solid). For ACSI: density of  $2(\theta-\mu) -u_j$ (dotted), $u_{(1)}$ or $u_{(n)}$ (dashed). For CSI: density of $u_j - 2(\theta-\mu)$ (dotted), $-|u|_{(n-1)}$ or $|u|_{(n-1)}$ (dashed). For all: calibration set size $n=79$, $\theta_i \equiv 1$, $\mu_i \equiv 1.2$, nominal level of confidence coverage $97.5\%$.} \label{fig:density}
\end{figure}

Figure \ref{fig:density} provides an illustration of these effects. Let $n = 79$ in all the cases. Let the nominal level of coverage be $97.5\%$, which corresponds to the choice of $k=1$ for the ACSI and $k=2$ for the CSI.

In the top row, we have IID $y_i \sim \text{N}(1,1)$ with $\theta_i \equiv 1$. Let $\mu_i \equiv 1.2$ exemplify a `tolerable' extent of misspecification, the density of $u_j = y_j -\mu_j$ is indicated by the solid line.  
\begin{itemize}[leftmargin=6mm,itemsep=0pt]
\item In the top-left plot regarding the ACSI, we show the density of $u_{(1)}$ and $u_{(n)}$ by the dashed lines, and the density of $2(\theta_j - \mu_j) - u_j$ by the dotted line. Since $y_i$ is symmetrically distributed around $\theta_i$, and the expectation of $u_j$ is always equal to that of $2(\theta_j - \mu_j) - u_j$, the distribution of $u_j$ is the same as that of $2(\theta_j - \mu_j) - u_j$. Clearly, the miscoverage \eqref{mis:ACSI} is $0$ in this case. 
\item In the top-right plot regarding the CSI, we show the density of $-|u_{(n-1)}|$ and $|u_{(n-1)}|$ by the dashed lines, and the density of $u_j - 2(\theta_j - \mu_j)$ by the dotted line. Notice the location shift due to the different expectations of $u_j$ and $u_j - 2(\theta_j - \mu_j)$, i.e. $\theta - \mu$ vs. $\mu -\theta$, which however does not cause miscoverage by Lemma \ref{lemma:IID:CSI}. 
\end{itemize}

Next, in the bottom row, we have IID $y_i \sim \text{Exp}(1)$ with $\theta_i \equiv 1$. Again, let $\mu_i \equiv 1.2$, and let the corresponding density of $u_j = y_j -\mu_j$ be indicated by the solid line. 
\begin{itemize}[leftmargin=6mm,itemsep=0pt]
\item In the bottom-left plot regarding the ACSI, we show the density of $u_{(1)}$ and $u_{(n)}$ by the dashed lines, where the density of $u_{(1)}$ is highly concentrated around its expectation. The density of $2(\theta_j - \mu_j) - u_j$ indicated by the dotted line flips around that of $u_j$, where the two have the same expectation. The long tail of $2(\theta_j - \mu_j) - u_j$ extends now to the left, with an appreciable probability mass beyond the probable values of $u_{(1)}$. This causes a large miscoverage \eqref{mis:ACSI}, where the confidence coverage level by the ACSI is only $86.3\%$ instead of the nominal level $97.5\%$.

\item In the bottom-right plot regarding the CSI, where the density of $-|u_{(n-1)|}$ and $|u_{(n-1)}|$ are given by the dashed lines, the density of $u_j - 2(\theta_j - \mu_j)$ indicated by the dotted line exhibits only a location shift to $u_j$. The resulting confidence coverage is $96.4\%$, still reasonably close to the nominal level $97.5\%$, because the location shift is small compared to the probable range $[-|u_{(n-1)|}, |u_{(n-1)}|]$.
\end{itemize}

The above example of Exp(1) shows that miscoverage by the ACSI may be much more severe than that of the CSI, when $y_i - \theta_i$ is far from being symmetrically distributed. In addition to building the best possible regression function $\mu(x)$, transforming the dependent variable for $y_i$ to have symmetrically distributed regression errors remains a desirable goal in practice, just like in classical regression inference.

\subsection{Design-based split-CSI} \label{sec:design}

In the above we have assumed a given (or pre-trained) regression function $\mu(x)$. Split-conformal intervals can easily accomplish the training that is needed in practice. First, one would use a random training set $s_1$ to obtain a regression function, denoted by $\mu(x, s_1)$, which yields value $\mu_i = \mu(x_i, s_1)$ for any unit $i\notin s_1$; second, one would use another random calibration set $s_2$ to obtain the shrinkage scores, which can be used to form the conformal confidence intervals for any holdout $\theta_j$, $j\notin s_1 \cup s_2$. 

\begin{center} 
\begin{tabular}{ll} 
\multicolumn{2}{l}{\textit{Algorithm-I. Design-based split-CSI conditional on reference set $s$}} \\ \toprule
\multicolumn{2}{l}{Input: $|s| =m$, $\{ (y_i, x_i): i\in s\}$, $n_1$, $\alpha = 1-\tfrac{k}{m - n_1}$, $\mu(x)$} \\
1: & select $s_1$ of size $n_1$ from $s$  by SRSWOR, obtain $\mu(x, s_1)$; \\
2: & select $j$ randomly from $s\setminus s_1$, let $s_2 = (s\setminus s_1) \setminus \{ j\}$, $n_2 = m-n_1 -1$; \\
3: & for each $i\in s\setminus s_1$, do \\
4: & \quad $\mu_i = \mu(x_i, s_1)$, $\dot{\theta}_i = \tfrac{1}{2} (y_i + \mu_i)$, $b_i = \tfrac{1}{2} (y_i - \mu_i)$; \\
5: & end for \\
6: & obtain the order statistics $|b|_{(1)}, ..., |b|_{(n_2)}$ of $\{ |b_i| : i\in s_2\}$; \\  
\multicolumn{2}{l}{Output: $\mathcal{C}_j(k) = \big[ \dot{\theta}_j - |b|_{(n_2 -k+1)},~ \dot{\theta}_j + |b|_{(n_2 -k+1)} \big]$} \\ \bottomrule
\end{tabular} 
\end{center}

Algorithm-I gives the split-sampling design, which yields the \emph{split-CSI}, conditional on the realised sample $\{ (y_i, x_i) : i\in s\}$. Similarly, one can obtain the split-ACSI $\mathcal{A}_j(k)$ using $b_{(k)}$ and $b_{(n_2 -k+1)}$ instead of $|b|_{(n_2 -k+1)}$. Notice that simple random sampling without replacement (SRSWOR) of $s_1$ from $s$ generates exchangeable joint distribution of $\{ (y_i, \mu_i, \dot{\theta}_i, b_i) : i\in s\setminus s_1\}$ by design. Moreover, the holdout unit $j$ is randomly selected from $s\setminus s_1$, which ensures the marginal prediction coverage of $y_j$ as given below. 

\begin{lemma} \label{lemma:SRS:pred}
For split-CSI by Algorithm-I, we have $\frac{1}{m} \sum_{j=1}^m \Pr\{ y_j \in \mathcal{C}_j(k) \} = 1 - \tfrac{k}{n_2 +1}$.
\end{lemma}

It is important to notice that the prediction coverage is marginal over all the units, because the specific indices of two distinct units are not exchangeable with each other. For instance, suppose $(y_1, y_2, y_3) = (1, 2, 1000)$ for three units $s\setminus s_1$, and $\theta_i = \theta = \mu_i = \mu$ for simplicity, then $y_3$ will not be covered by $C_3(k=1)$ because $|b_3|$ is much larger than $|b_1|$ and $|b_2|$ when $s_2 = \{ 1, 2\}$, whilst $y_1$ or $y_2$ will be covered by $\mathcal{C}_1(k=1)$ or $\mathcal{C}_2(k=1)$, such that the coverage is $1-\tfrac{1}{3} = \tfrac{2}{3}$ marginally but not for any specific unit.   

One can implement sample-splitting for each in-sample unit $j$, yielding $s_1 \cup s_2 = s \setminus \{ j\}$, with the respective subsample sizes $n_1$ and $n_2$. In this way, all the CSIs can be obtained over $m$ splits. The sample space of $(s_1, s_2)$ by such conditional split-sampling given the holdout $j$ is the same as that of unconditional split-sampling by Algorithm-I, which happens to yield the same $j$, and each distinct $(s_1, s_2)$ have the same sampling probability. 

When it comes to the choice of $(n_1, n_2)$, it is common in the literature of conformal prediction to let $n_1 \approx n_2$ (e.g. Lei et al., 2018). The choice needs to balance between the training efficiency of $\mu(x, s_1)$, which favours larger $n_1$, and the scoring efficiency of $|b|_{(n_2 -k+1)}$, which favours larger $n_2$. Since the same contrasting forces exist for the design-based split-CSI, \emph{half-sample training} remains the default choice. 

Now, for confidence inference of $\{ \theta_i : i\in s\}$ conditional on $s$, we have the following result for the split-CSI regardless the distribution $\mathcal{P}$ of $\{ (y_i, x_i) : i\in s\}$.

\begin{lemma} \label{lemma:SRS:CSI}
The split-CSI by Algorithm-I has exactly the marginal confidence coverage \eqref{cvr_conf} with respect to the split-sampling design, where $\alpha = 1- \tfrac{k}{n_2 +1}$, (i) if $\mu(x_i, s_1) = \theta_i$ for $i\in s\setminus s_1$, or (ii) if $y_i -\theta_i \overset{\text{D}}{\simeq} \theta_i -y_i$ under the split-sampling design. Otherwise, the marginal miscoverage probability is given by
\begin{equation} \label{mis:SRS:CSI}
\epsilon_{\mathcal{C}}^{\text{split}} = \frac{1}{m} \sum_{j=1}^m \Pr\big\{ -|u|_{(n_2 -k+1)} \leq u_j - 2 (\theta_j - \mu_j) \leq |u|_{(n_2 -k+1)} \big\} - \alpha ~.
\end{equation} 
\end{lemma}

Although a similar result of miscoverage can be given for the split-ACSI, we omit the details here because, as the discussions in Sections \ref{sec:mis_cvr_1} and \ref{sec:mis_cvr_2} show, the coverage of ACSI can easily deteriorate as the distribution of $y_i -\theta_i$ deviates from symmetry. In other words, the ACSI is not a generally viable method for confidence inference, despite the insight it brings to the underlying mechanism.

\subsection{Miscoverage (II)} \label{sec:mis_cvr_2}

Both the miscoverage \eqref{mis:SRS:CSI} of split-CSI and the interval width would be reduced, if one manages to reduce the discrepancies $|\theta_i - \mu_i|$ by the choice of $\mu(x)$, which is a primary objective for regression. The result below is instructive in this respect. 
 
\begin{coro} \label{coro:null}
For confidence inference of the marginal expectation $\theta = \mathbb{E}(y_i)$ with respect to $\mathcal{P}$, the split-CSI using $\mu(x_i, s_1) = \tfrac{1}{n_1} \sum_{i\in s_1} y_i$ attains the marginal confidence coverage \eqref{cvr_conf} asymptotically, as $n_1 \rightarrow \infty$.
\end{coro}

Notice that the target of confidence inference is the marginal expectation $\theta$ of $y_i$, not the conditional expectation $\theta(x_i)$. The null model with $x_i \equiv 1$, or $\mu_i \equiv \mu$, is correctly specified in this respect, and one would estimate $\mu$ simply by the training set mean $\bar{y}(s_1) = \tfrac{1}{n_1} \sum_{i\in s_1} y_i$. Since $\bar{y}(s_1)$ converges to $\theta$ in probability with respect to SRSWOR, as $n_1 \rightarrow \infty$ asymptotically, $u_j - 2 (\theta_j - \mu_j)$ converges to $u_j$ in distribution, such that the miscoverage \eqref{mis:SRS:CSI} tends to $0$. 

Let us illustrate this by revisiting the example in Section \ref{sec:mis_cvr_1}. Let the sample $s$ consist of $m=80$ units. Let $n_1 = 40$ for half-sample training, such that the nominal level of coverage is $0.95$ with $k=1$ for the split-ACSI and $k=2$ for the CSI. 

Let $\theta \equiv 1$ in all the cases to be considered. In addition to $\text{N}(1,1)$ and $\text{Exp}(1)$, we consider $\text{logN}(-0.5,1)$, as well as two mixture samples where the first $40$ $y_i$'s are drawn from either $\text{N}(1,1)$ or $\text{logN}(-0.5,1)$ and the last $40$ $y_i$'s are drawn from $\text{Exp}(1)$. Notice that the two mixture samples do not satisfy the assumption of IID or exchangeability for model-based conformal prediction inference.  

\begin{table}[ht]
\centering
\caption{Monte Carlo coverage of split-CSI or split-ASCI.}
\begin{tabular}{ccccc} \toprule
 & \multicolumn{2}{c}{Prediction coverage} & \multicolumn{2}{c}{Confidence coverage} \\
Distribution $\mathcal{P}$ & Split-CSI & Split-ACSI & Split-CSI & Split-ASCI \\ \hline 
IID $\text{N}(1,1)$ & 0.950 & 0.950 & 0.950 & 0.958 \\   
IID $\text{logN}(-0.5,1)$ & 0.950 & 0.950 & 0.950 & 0.900 \\ 
IID $\text{Exp}(1)$ & 0.950 & 0.950 & 0.950 & 0.843 \\ 
Half $\text{N}(1,1)$, half $\text{Exp}(1)$ & 0.950 & 0.950 & 0.952 & 0.943 \\ 
Half $\text{logN}(-0.5,1)$, half $\text{Exp}(1)$ & 0.950 & 0.950 & 0.951 & 0.899 \\ \bottomrule 
\end{tabular} \label{tab:coro}
\end{table}

Table \ref{tab:coro} shows the Monte Carlo results, where we simulate a sample of $80$ $y$-values in each case. Given the realised sample, we repeat the split-sampling $10^3$ times to calculate the marginal prediction and confidence coverages of the split-CSI and split-ACSI, respectively; the number of repetitions is sufficient for significant results to the 3rd digit in Table \ref{tab:coro}. 

It can be seen that the confidence coverage of the split-ACSI decays quickly in the case of $\text{logN}(-0.5,1)$ or $\text{Exp}(1)$ due to the asymmetric distribution of $e_i = y_i -\theta_i$, whereas the coverage in each mixture case---$0.943$ or $0.899$---falls between those of the two `component' cases, i.e. $0.958$ and $0.843$, or $0.900$ and $0.843$.

In contrast, the confidence coverage of the split-CSI is nearly equal to the nominal level $95\%$ in all the IID cases, although $n_1 =40$ is far from the large-sample setting of Corollary \ref{coro:null}. In both the mixture samples, the split-sampling design for in-sample confidence inference yields near-nominal level of coverage by the split-CSI.

Of course, the null model with $\mu_i \equiv \mu$ is not really useful for regression. The example here serves only to provide some insight to the robustness of design-based split-CSI. In practice, one should focus on limiting the discrepancies $|\theta_i - \mu_i|$, to achieve near-nominal level of coverage and to reduce the width of split-CSI.

\section{Application} \label{sec:application}

The Italian National Institute for Statistics, Istat, recently moved from the traditional decennial census to integrated use of administrative data and annual sample surveys to produce census-related statistics. We consider the commuting statistics in this application. Due to the limited sample size of the relevant survey, SAE techniques are needed to produce statistics at detailed geographical levels. 

\subsection{Data and results}

In terms of the setup introduced in Section \ref{sec:introduction}, let $\mathcal{D}$ contain the sampled municipalities (or areas henceforth) in the region Tuscany, where $m=164$. For the FH-model \eqref{FHmod}, let the area parameter $Y_i$ be the ratio of commuters to the area population size in 2021 according to the Population Register. Let the outcome variable $y_i$ be the direct estimate of $Y_i$ based on the Italian population permanent census survey in 2021, where $\theta_i = \mathbb{E}(y_i \mid Y_i) = Y_i$ with respect to finite population sampling. The feature $x_i$ used below is the ratio between the 2011 census commuter count and the 2021 area population size; we use the latter as the denominator since the Population Register was not as well-developed in 2011 as in 2021. 

\begin{figure}[ht]
\centering
\includegraphics[width=13cm,height=5cm]{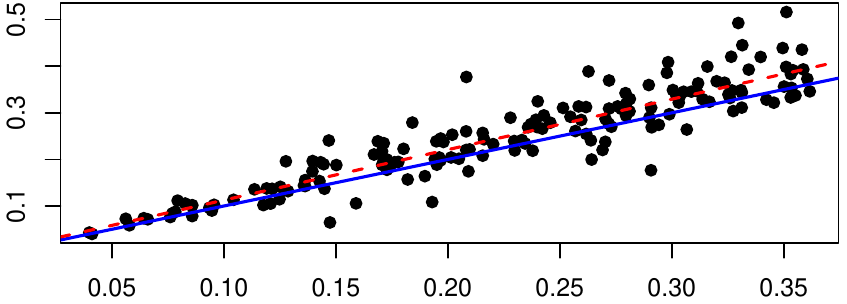}
\caption{Direct estimate vs. feature. OLS slope $\beta$ (dashed), $\beta = 1$ (solid).} \label{fig:xy}
\end{figure}

\begin{table}[ht] 
\centering
\caption{Summary statistics of $y_i$, $x_i$, and CV of $y_i$ over 164 areas.} 
\begin{tabular}{c|c|c|c|c|c|c} \toprule
 & Min.  & 1st Q. & Median & Mean & 3rd Q. & Max.  \\   \hline
$y_i$ & 0.0401 & 0.1868 & 0.2558 & 0.2509 & 0.3295 & 0.5163 \\
$x_i$ & 0.0399 & 0.1683 & 0.2375 & 0.2276 & 0.2993 & 0.3616 \\ 
CV (\%) & 2.180 & 5.021 & 6.744 & 7.700 & 9.302 & 20.650 \\ \bottomrule
\end{tabular} \label{tab:summary}
\end{table}

The scatter plot of $(x_i, y_i)$ is given in Figure \ref{fig:xy}. Adopting the regression function $\mu(x_i) = x_i \beta$ appears reasonable, which is not surprising since the commuter pattern changes rather slowly over time. Table \ref{tab:summary} shows in addition the summary statistics of $y_i$, $x_i$, and the coefficient of variation (CV) of $y_i$ in percentage.
 
\begin{figure}[ht]
\centering
\includegraphics[scale=0.42]{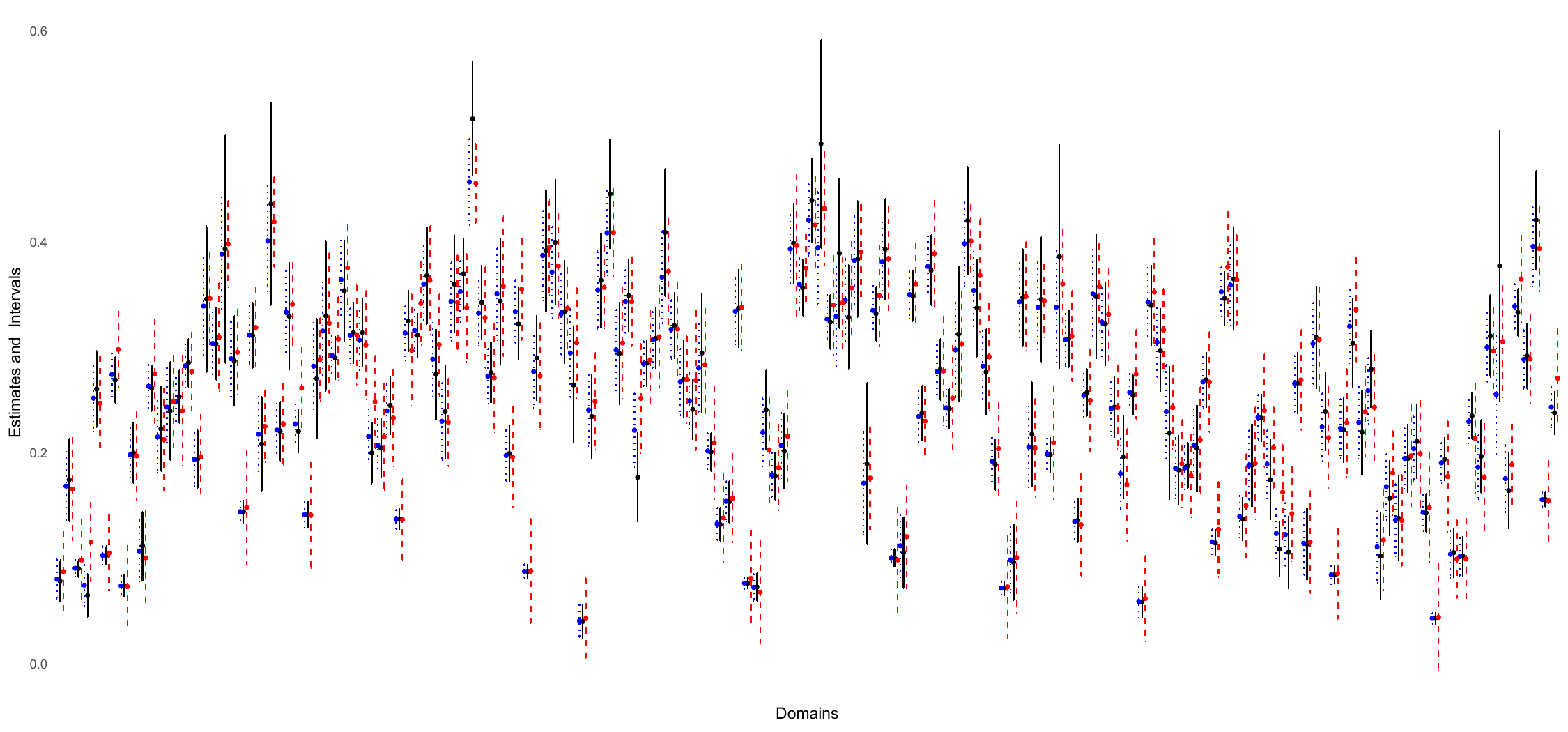} \\
\caption{$95\%$ confidence intervals with marked centres, direct (solid black), EBLUP (dotted blue), split-CSI (dashed red).} \label{fig:CI}
\end{figure}

Figure \ref{fig:CI} shows the confidence intervals with the nominal level $95\%$. The direct interval is centred on $y_i$, using its estimated SE (corresponding to CV in Table \ref{tab:summary}) and assuming $y_i \sim \text{N}(\theta_i, \text{SE}_i)$. The EBLUP interval is centre on the EBLUP, $\hat{\theta}_i^H = x_i \hat{\beta} + \hat{v}_i$, using its estimated MSE under the model \eqref{FHmod} and assuming $\hat{\theta}_i^H - \theta_i \sim \text{N}(0, \text{MSE}_i)$. The CSI is given by Algorithm-I with the reference set $s = \mathcal{D}$ and $n_1 = m/2$, using $\mu(x_i) = x_i \beta$, which does not require any variance or MSE estimation.   

The split-CSI and EBLUP intervals are more closely aligned with each other, both of which pull the more outlying direct intervals towards the rest confidence intervals. This is because the direct interval is centred on $y_i$, where $\{ y_i : i\in s\}$ are over-dispersed compared to $\{ \theta_i \}$ due to the extra sampling variances.  

The interval half-width averaged over all the 164 areas is 0.036 for the direct intervals, 0.029 for the EBLUP intervals, and 0.046 for the CSIs. The split-CSI has a more similar halfwidth across the areas, which can be much wider than the narrowest direct (or EBLUP) intervals, such as in some areas with the smallest $y_i$-values, but much narrower than the widest direct intervals, such as  in some areas with the largest $y_i$-values.

\subsection{Method evaluation} \label{sec:evaluation}

We conduct finite population simulation for method evaluation (Tzavidis et al., 2018, Sec. 4.2). The aim is to facilitate the choice of method, by examining their performance over repeated sampling from similar finite populations. The main concerns include the coverage and efficiency of a given interval estimator, the potential sensitivity of certain methodological elements, such as the sampling variance estimator needed for direct and EBLUP intervals, the misspecification effects of $\mu(x)$ for the EBLUP and CSI, and the choice of $(n_1, n_2)$ for the split-sampling design of CSI.    

\subsubsection{Coverage and efficiency} 

The basic setup of design-based simulation is as follows.
\begin{itemize}[leftmargin=6mm,itemsep=0pt]
\item Let the synthetic population parameters be given as 
\[
Y_i^*(w) = w y_i + (1-w) x_i \quad\text{where}\quad i=1, ..., m.
\]
Different values of the coefficient $w$ are explored, such as $w= 0.1, 0.3, 0.5, 0.7, 0.9$, where the prediction power of $x_i$ increases as $w$ reduces from $1$ to $0$. 
  
\item Given each set of $Y_i^*(w)$ indexed by $w$, we simulate the direct estimates $L$ times. For each $l=1, ..., L$, the direct estimates are generated by $y_i^{(l)} \overset{\text{IID}}{\sim} \text{N}(Y_i^*, \psi_i)$, for $i=1, ..., m$, where $\psi_i$ is the variance of $y_i$ in the actual survey reported earlier. 

\item Base on each sample $\{ y_i^{(l)} : i=1, ..., m\}$, we calculate the confidence intervals as in the actual survey. 
\end{itemize}

Notice that sampling from a given synthetic finite population as above does not replicate the design of the actual survey, because the relevant sampling frame and the details of the sampling design are not available to us. The simulation approach here is convenient for secondary analysts in similar situations. Potential departures from the normal distribution of $y_i$ will be explored in sensitivity analysis later.  

Notice also that for the direct and EBLUP intervals, it is necessary to estimate the sampling variance $\psi_i$ in practice. We consider two options: (i) use the IID model variance estimator $y_i (1-y_i)/n_i$, which ignores the actual design effect; (ii) adjust the IID variance by the design-effect, given as $\psi_i /\{ Y_i^* (1- Y_i^*)/n\}$, which yields approximately unbiased variance estimation over repeated simulations.

\begin{table}[ht] 
\centering
\caption{Marginal confidence coverage ($\%$) by synthetic population ($w$), $\mu(x_i) = x_i \beta$. Unbiased sampling variance $\hat{\psi}_i$ or IID model variance. Nominal level $95.1\%$. Average interval half-width ($\%$) in parentheses. Based on $L=1000$ simulations.} 
\begin{tabular}{c|c|cc|cc} \toprule
& Split-CSI & \multicolumn{2}{c|}{Direct}& \multicolumn{2}{c}{EBLUP}\\
$w$ & $m = 2 n_1$ & $\hat{\psi}_i$ & IID & $\hat{\psi}_i$ & IID \\ \hline
0.1 & 95.1 (2.3) & 94.1 (3.6) & 63.5 (1.6) & 98.8 (1.4) & 69.1 (1.4)\\
0.3 & 95.1 (2.6) & 94.2 (3.6) & 63.8 (1.6) & 86.3 (1.5) & 67.4 (1.4)\\ 
0.5 & 95.1 (3.1) & 94.5 (3.6) & 64.1 (1.6) & 88.0 (2.1) & 66.1 (1.5)\\  
0.7 & 95.1 (3.8) & 94.6 (3.6) & 64.4 (1.6) & 91.0 (2.6) & 65.7 (1.5) \\ 
0.9 & 95.1 (4.6) & 94.8 (3.6) & 64.5 (1.6) & 92.2 (2.9) & 65.4 (1.6) \\ \bottomrule
\end{tabular} \label{tab:cvr_hw}
\end{table}

Table \ref{tab:cvr_hw} reports the marginal confidence coverage \eqref{cvr_conf} of the different intervals over repeated sampling from the finite populations. The average half-width of the intervals over the $m$ areas is given in the parentheses. The half-width of the direct intervals depends only on the sampling variance estimate, which is the same regardless $w$ in the simulation setup here. The half-width of the EBLUP intervals  depends on the estimated variance of the random effects in the FH model \eqref{FHmod} in addition, which increases with $w$. The half-width of the CSI depends on $|\theta_i - \mu_i|$, which increases with $w$; notice that these conformal intervals do not require variance or MSE estimation at all.

As can be expected from Lemma \ref{lemma:SRS:CSI}, the split-CSI achieves the nominal coverage level in all the populations. In the case of unbiased $\hat{\psi}_i$, the direct interval achieves near-nominal levels of converge, whereas the coverage level of the EBLUP interval is erratic across the different populations, and its coverage is not necessarily closer to the nominal level as the discrepancies $|\theta_i -\mu_i|$ decrease with $w$. The coverage of either the direct or EBLUP interval collapses dramatically when one uses the IID model variance that ignores the actual design effects. 

It is worthwhile to point out that, since variance estimation is necessary for the direct interval, and model estimation (including both the regression function and the variance components) is necessary for the EBLUP interval, these methods can never attain exactly the nominal level of coverage in practice, as illustrated here. 

When it comes to the efficiency of interval estimation, the half-width of the direct interval does not depend on $w$ here, while the half-width of either the EBLUP interval or the split-CSI decreases as the prediction power of $x_i$ increases, i.e. as the discrepancies $|\theta_i - \mu_i|$ decrease. In the current setup, the half-width is about $50\%$ larger by the split-CSI than by the EBLUP interval. However, given the erratic coverage of the EBLUP interval, such gains of efficiency by the model-based interval can be misleading in reality. The half-width of the split-CSI is smaller than that of the direct interval if $w <0.7$ but larger if $w \geq 0.7$. This is not surprising since, as $w \rightarrow 1$, the CSI would need to accommodate an increasing discrepancy $|\theta_i - \mu_i|$, while the direct interval is always centred on $y_i$ regardless $w$. 

We notice that it is possible to consider other plausible efficiency measures than the half-width in simulation studies.  For instance, one may consider the root of average MSE of the interval bounds, given as
\[
\tau_i = \sqrt{\tfrac{1}{2} \{ \mathbb{E}[(\text{UB}_i - \theta_i)^2] + \mathbb{E}[(\text{LB}_i - \theta_i)^2] \}} 
\]
where $\text{LB}_i$ and $\text{UB}_i$ refer to the lower and upper bounds of the interval, respectively. Since all the intervals here are given as $\tfrac{1}{2} (\text{LB}_i + \text{UB}_i) \pm \text{Half-width}_i$, $\tau_i$ would capture both the interval half-width and the variance of the interval centre. The details are omitted here, which show the same pattern as Table \ref{tab:cvr_hw}, except they favour more the split-CSI and penalise the direct interval harder, due to the variances of the respective interval centres, which is $y_i$ for the direct interval,  $0.5(\hat{\mu}_i + y_i)$ for the split-CSI, and $\hat{\mu}_i + \hat{v}_i$ for the EBLUP interval.

\subsection{Sensitivity analysis}

Here we explore several relevant issues by means of sensitivity analysis. 

\paragraph{Misspecifying $\mu(x)$}
Robustness against misspecified regression function $\mu(x)$ is a key motivation for conformal methods. Although $\mu(x_i) = x_i \beta$ looks reasonable for these data (Figure \ref{fig:xy}), we can explore its potential misspecification effects by comparisons to the coverage of the confidence intervals using the null model $\mu_i \equiv \mu$. 

\begin{table}[ht] 
\centering
\caption{Marginal confidence coverage ($\%$) by synthetic population ($w$), using the null model $\mu_i \equiv \mu$. Sampling variance by unbiased $\hat{\psi}_i$. Nominal  level $95.1\%$. Average interval half-width ($\%$) in parentheses. Based on $L=1000$ simulations.} 
\begin{tabular}{c|c|c} \toprule
$w$ & Split-CSI, $m = 2 n_1$ & EBLUP, $\hat{\psi}_i$ \\ \hline
0.1 & 95.1 (8.2) & 95.4 (3.5) \\
0.3 & 95.1 (8.4) & 95.4 (3.5) \\ 
0.5 & 95.1 (8.8) & 95.4 (3.5) \\ 
0.7 & 95.1 (9.1) & 95.0 (3.5) \\ 
0.9 & 95.1 (9.5) & 95.4 (3.5) \\ \bottomrule
\end{tabular} \label{tab:cvr_null}
\end{table}

Table \ref{tab:cvr_null} shows that the marginal confidence coverage of the null-model EBLUP interval barely varies with $w$, and is closer to the nominal level than in Table \ref{tab:cvr_hw}. This is because using the null model removes the misspecification of $\mu(x_i)$, where the random effects in the FH model \eqref{FHmod} can account for the differences from one $\theta_i$ to another. 

It follows from the EBLUP results in Table \ref{tab:cvr_null} and Table \ref{tab:cvr_hw} that, (a) the linear function $\mu(x_i) = x_i \beta$ must be somewhat misspecified for the synthetic populations, as well as the target population in 2021, although one would not reject the linear function based on a significance test, and (b) the misspecifcation does cause erratic miscoverage by the EBLUP interval, now that near-nominal levels of coverage could have been achieved by the null model and the associated MSE estimation. 

In contrast, the split-CSI achieves the nominal level of marginal confidence coverage \eqref{cvr_conf} regardless the choice of $\mu(x)$, as can be expected from the results developed in this paper,  which aligns with the intention of conformal methods. A good regression function is nevertheless helpful for improving the efficiency of interval estimation, as one can see by comparing the half-widths of split-CSI in Table \ref{tab:cvr_null} and Table \ref{tab:cvr_hw}.

\paragraph{Choice of split-sampling design} The default split-sampling design uses half-sample training. We can explore other choices for sensitivity analysis. 

\begin{table}[ht] 
\centering
\caption{Coverage and average half-width of split-CSI by $n_1/m$, all numbers in $\%$, for synthetic population $w=0.3, 0.7$. Nominal level $95.1\%$. Based on $L=1000$ simulations.} 
\begin{tabular}{c|cccccc} \toprule
& \multicolumn{6}{c}{$w=0.3$} \\ 
$n_1/n$ & 15 & 30 & 40 & 60 & 75 & 85 \\ \hline 
Coverage & 95.1 & 95.1 & 95.1 & 95.1 & 95.0 & 95.0 \\ 
Half-width & {2.6} & {2.6} & {2.6} & {2.6} & {2.7} & {3.1} \\ \bottomrule
& \multicolumn{6}{c}{$w=0.7$} \\ 
$n_1/n$ & 15 & 30 & 40 & 60 & 75 & 85 \\ \hline 
Coverage & 95.1 & 95.1 & 95.1 & 95.1 & 95.1 & 95.1 \\ 
Half-width & {3.8} & {3.8} & {3.8} & {3.9} & {4.0} & {4.4} \\ \bottomrule
\end{tabular} \label{tab:n1}
\end{table}

The results for the synthetic population with $w=0.3$ or 0.7 are given in Table \ref{tab:n1}, where the split-CSI uses $\mu_i = x_i \beta$. It is clear that the different ratios of sample-split have only a small effect on the interval width, which is fairly stable as $n_1/m$ increases from $15\%$ to $50\%$ and starts to increase with $n_1/m$ from then on. The default choice of half-sample training seems reasonable, as far as the average interval width is concerned.

\paragraph{Departure from normal distribution} We have simulated $y_i \sim \text{N}(Y_i^*, \psi_i)$ above since we cannot replicate the actual sampling design here. To explore the sensitivity of this choice, we can instead simulate log-normal area total estimate $t_i$ to yield $y_i = t_i/N_i$, given
\[
t_i \sim \text{logN}(\tau_i, \sigma_i^2) \quad\text{with}\quad \tau_i = \log (N_i Y_i^*) - \tfrac{1}{2} \sigma_i^2 
\quad\text{and}\quad \sigma_i^2 = \log \big( 1 + \psi_i/(Y_i^*)^2 \big)
\] 
where $N_i$ is the area population size, $\mathbb{E}(t_i) = N_i Y_i^*$ and $\mathbb{V}(t_i) = N_i^2 \psi_i$.  

\begin{table}[ht] 
\centering
\caption{Marginal confidence coverage ($\%$) by synthetic population ($w$), $\mu(x_i) = x_i \beta$. Unbiased sampling variance $\hat{\psi}_i$ or IID variance. Nominal level $95.1\%$. Average interval half-width ($\%$) in parentheses. Based on $L=1000$ simulations of log-normal area total estimate $t_i$.} 
\begin{tabular}{c|c|cc|cc} \toprule
& Split-CSI & \multicolumn{2}{c|}{Direct}& \multicolumn{2}{c}{EBLUP}\\
$w$ & $m = 2 n_1$ & $\hat{\psi}_i$ & IID & $\hat{\psi}_i$ & IID \\ \hline
0.1 & 94.9 (2.3) & 94.2 (3.5) & 63.6   (1.6) & 98.9 (1.4) & 68.9  (1.4)  \\
0.3 & 94.8 (2.6) & 94.4 (3.5) & 63.9   (1.6) & 86.2 (1.4) & 67.3  (1.4) \\ 
0.5 & 95.0 (3.1) &  94.6 (3.5) & 64.1 (1.6) & 
87.7 (2.1) & 66.1 (1.5) \\  
0.7 & 95.0 (3.8) &  94.8 (3.5) & 64.5 (1.6) & 91.1 (2.6)& 65.7 (1.5)\\ 
0.9 & 95.0 (4.6) &  94.9 (3.6) & 64.7 (1.6) & 92.4 (2.9)& 65.5 (1.6) \\ \bottomrule
\end{tabular} \label{tab:cvr_logN} \\
 \end{table}

The differences to the results in Table \ref{tab:cvr_hw} are quite small for all the intervals, because the log-normal skewness is small with these data and it does not cause a large difference to the normal distribution. For comparison, the logN-distribution skewness of $t_i/N_i$ here is much smaller than that of $\text{logN}(-0.5,1)$ exemplified in Section \ref{sec:mis_cvr_2}.

\subsection{Conclusion}

Based on the results of method evaluation, we conclude that the design-based split-CSI yields the most reliable marginal confidence coverage \eqref{cvr_conf} for the Italian population permanent census survey in 2021. The split-sampling design is agnostic to the sample distribution of $\{ (y_i, x_i) : i\in s\}$, and the conditional confidence inference property of the split-CSI is attractive theoretically as well as practically. In comparison, the apparent efficiency of the EBLUP interval is likely to  mislead with respect to finite population inference, whereas the direct interval may be less efficient and its coverage can never achieve exactly the nominal level due to the need of variance estimation.    

It is worthwhile to add that the areas considered in this application were actually selected according to a complex stratified rotating sample design, where $95$ of the areas are selected with probabilities less than $1$ and the rest are self-representing areas selected with probability $1$. In other words, the assumption of IID $(y_i, x_i)$ or exchangeable $\{ (y_i, x_i) : i\in s\}$ for model-based conformal inference likely does not hold. To be sure, we can confirm that the coverage of the split-CSI remains valid, when we limit the reference set to the $95$ areas or the $69$ areas, separately, the details of which are omitted here.  

Moreover, the split-CSI is easier to implement than the direct or EBLUP interval, because there is no need at all to estimate the sampling variances $\psi_i$. Although the FH-model \eqref{FHmod} assumes these to be known, the estimation of $\psi_i$ is a perennial issue in practice. To improve the efficiency of EBLUP, it is common to smooth the direct variance estimates, which will cause bias of the resulting estimator of $\psi_i$ and miscoverage of the EBLUP interval. Even for the direct interval, unbiased variance estimation may be unrealistic where the area sample sizes are small and the direct estimator $y_i$ is not simply the estimator of Horvitz and Thompson (1952). Conformal confidence inference by the split-CSI avoids the matter completely, allowing one to focus on the regression function $\mu(x)$, in order to achieve the nominal coverage level and to reduce the interval width. For the Italian population permanent census survey, one may improve the regression function we have used by bringing in additional features.

\section{Final remarks} \label{sec:final}

Barber et al. (2021) show that the validity of conformal prediction inference holds marginally, and no conformal method can guarantee the prediction coverage of any particular out-of-sample outcome. The marginal interpretation applies as well to conformal confidence inference \eqref{cvr_conf}, which is our target in this paper.

We have developed a general approach to in-sample conformal confidence inference, which fills a gap between classical regression and conformal inference. The conformal shrinkage interval (CSI) by any shrinkage constant, $\phi \in (0,1)$, can achieve the nominal level of prediction coverage for the in-sample outcomes, while the shrinkage constant $\phi =0.5$ in particular can be used for confidence inference generally. The split-sampling design and the resulting split-CSI are proposed for confidence inference conditional the actual sample, which is agnostic to the underlying sample-data distribution that is unknown and may not satisfy the IID or exchangeability assumption. The conditions for the exact confidence coverage are identified, and the miscoverage probability is given otherwise.  

The proposed conformal confidence inference is applied to the Italian permanent census survey 2021. Method evaluation shows that the split-CSI yields the most reliable confidence coverage --and it is easier to implement-- than the alternative of direct or model-based interval. Conformal confidence inference provides thus a robust primary choice for SAE; the challenge for the practitioners will be to potentially improve the efficiency of the split-CSI without sacrificing its coverage level. This belongs to an active area of research, where one attempts to improve feature-specific properties of the conformal intervals without sacrificing the marginally valid coverage. In the context of SAE, a relevant question is how to differentiate the areas according to the variances of their direct estimators, even when one only have the variance estimates but not the true variances.

\appendix
\section{Proofs} \label{sec:proof}

Proof of Lemma \ref{lemma:IID:pred}.
\begin{proof}
Given IID $\big( y_i, \dot{\theta}_i(\phi), b_i(\phi), |b_i|(\phi) \big)$ for $i\in \{1, ..., n\}\cup \{ j\}$, and feasible $k$, we have
\begin{align*}
\Pr\big( y_j \in \mathcal{A}_j(k; \phi)\big) & = \Pr\{ b_{(k)}(\phi) \leq y_j - \dot{\theta}_j(\phi) \leq b_{(n-k+1)}(\phi) \} \\
& = \Pr\{ b_{(k)}(\phi) \leq  b_j(\phi) \leq b_{(n-k+1)}(\phi) \} = 1 - \tfrac{2k}{n+1} ~.
\end{align*}
and
\begin{align*}
\Pr\big( y_j \in \mathcal{C}_j(k; \phi)\big) & = \Pr\{ -|b|_{(n-k+1)}(\phi) \leq y_j - \dot{\theta}_j(\phi) \leq |b|_{(n-k+1)}(\phi) \} \\
& = \Pr\{ |b_j|(\phi) \leq |b|_{(n-k+1)}(\phi) \} = 1 - \tfrac{k}{n+1} ~. \qedhere
\end{align*}
\end{proof}

\noindent
Proof of Lemma  \ref{lemma:IID:ACSI}.
\begin{proof} Let $b_i(\phi) = (1-\phi) u_i$. For any feasible $k$, we have
\begin{align*}
\mathcal{A}_j(k; \phi) & = \big[ \phi y_j + (1- \phi) \mu_j + (1-\phi) u_{(k)},~ \phi y_j + (1- \phi) \mu_j + (1-\phi) u_{(n-k+1)} \big] \\
& = \big[ \phi y_j + (1- \phi) (\mu_j + u_{(k)}),~ \phi y_j + (1- \phi) (\mu_j + u_{(n-k+1)}) \big] ~.
\end{align*}
Given $\phi = 0.5$, prediction coverage refers to the probability of
\begin{gather*}
y_j \in \mathcal{A}_j(0.5; \mu) \quad\Leftrightarrow\quad u_{(k)} \leq u_j \leq u_{(n-k+1)} \\
\quad\Leftrightarrow\quad (\mu_j -\theta_j) + u_{(k)} \leq e_j \leq (\mu_j -\theta_j) + u_{(n-k+1)}
\end{gather*}
where $e_j = y_j - \theta_j$, and confidence coverage refers to the probability of
\begin{gather*}
\theta_j \in \mathcal{A}_j(0.5; \mu) \quad\Leftrightarrow\quad 
y_j + \mu_j + u_{(k)} \leq 2 \theta_j \leq  y_j + \mu_j + u_{(n-k+1)} \\
\quad\Leftrightarrow\quad u_{(k)} \leq 2(\theta_j -\mu_j) - u_j \leq u_{(n-k+1)} \\
\quad\Leftrightarrow\quad (\mu_j -\theta_j) + u_{(k)} \leq -e_j \leq (\mu_j -\theta_j) + u_{(n-k+1)} ~.
\end{gather*}
This yields $\epsilon_{\mathcal{A}}$ by \eqref{mis:ACSI}, which becomes zero given symmetrically distributed $e_i$.
\end{proof}

\noindent
Proof of Lemma  \ref{lemma:IID:CSI}.
\begin{proof}
For any feasible $k$, we have $|b|_{(n-k+1)} = \tfrac{1}{2} |u|_{(n-k+1)}$ given $\phi =0.5$, such that
\begin{gather*}
\theta_j \in \mathcal{C}_j(k;0.5) \quad\Leftrightarrow\quad
\mu_j - |u|_{(n-k+1)} + e_j \leq \theta_j \leq \mu_j + |u|_{(n-k+1)} + e_j \\
\Leftrightarrow\quad -|u|_{(n-k+1)} \leq u_j - 2 (\theta_j - \mu_j) \leq |u|_{(n-k+1)}
\end{gather*}
where $e_j = y_j - \theta_j$, and
\begin{gather*}
y_j \in  \mathcal{C}_j(k;0.5) \quad\Leftrightarrow\quad
\mu_j - |u|_{(n-k+1)} \leq \theta_j + e_j \leq \mu_j + |u|_{(n-k+1)} \\
\Leftrightarrow\quad -|u|_{(n-k+1)} \leq u_j \leq |u|_{(n-k+1)} ~.
\end{gather*}
This yields $\epsilon_{\mathcal{C}}$ by \eqref{mis:CSI}, and the result (i). Next, we have $y_j -\theta_j \overset{D}{\simeq} \theta_j - y_j$ or, equivalently, $y_j \overset{D}{\simeq} 2\theta_j - y_j$ if $y_j$ is symmetrically distributed around $\theta_j$, in which case
\begin{gather*}
\theta_j \in \mathcal{C}_j(k;0.5) \quad\Leftrightarrow\quad -|u|_{(n-k+1)} \leq y_j - 2 \theta_j + \mu_j \leq |u|_{(n-k+1)} \\
\Leftrightarrow\quad \mu_j -|u|_{(n-k+1)} \leq 2 \theta_j - y_j \leq \mu_j + |u|_{(n-k+1)} \\
\Leftrightarrow\quad \mu_j -|u|_{(n-k+1)} \leq y_j \leq \mu_j + |u|_{(n-k+1)} \\
\Leftrightarrow\quad -|u|_{(n-k+1)} \leq u_j \leq |u|_{(n-k+1)} \quad\Leftrightarrow\quad y_j \in  \mathcal{C}_j(k;0.5) ~. \qedhere
\end{gather*}
\end{proof}

\noindent
Proof of Lemma \ref{lemma:SRS:pred}.
\begin{proof}
The split-sampling design implies that the joint distribution of $\{ |b_i| :  i\in s\setminus s_1 \}$ is exchangeable by design, where $b_i = 0.5 u_i$ using the shrinkage constant $\phi =0.5$ and $u_i = y_i - \mu(x_i, s_1)$. Given feasible CSI $\mathcal{C}_j(k)$, we have, for a randomly selected unit $j$ from the $n_2 +1$ units in $s\setminus s_1$, and the corresponding calibration set $s_2 = s_1 \setminus \{ j\}$,  
\begin{align*}
\Pr\big\{ y_j \in \mathcal{C}_j(k)\big\} & = \Pr\{ |b_j| \leq |b|_{(n_2-k+1)} \} \\
& = \frac{1}{n_2 +1} \sum_{i=1}^{n_2 +1} \Pr\{ |b_i| \leq |b|_{(n_2-k+1)} \}  = 1 - \tfrac{k}{n_2 +1}
\end{align*}
where the first quality holds by definition, the second equality follows because $j$ is randomly selected from $s\setminus s_1$, and the third equality follows because each term in the sum is identical due to exchangeability. Over repeated split-sampling, we obtain
\[
\frac{1}{m} \sum_{j=1}^m \Pr\big\{ y_j \in \mathcal{C}_j(k)\big\} = 1 - \tfrac{k}{n_2 +1} ~. \qedhere
\]
\end{proof}

\noindent
Proof of Lemma \ref{lemma:SRS:CSI}.
\begin{proof} By the same argument in Lemma \ref{lemma:IID:CSI}, we obtain 
\begin{gather*}
\theta_j \in \mathcal{C}_j(k) \quad\Leftrightarrow\quad -|u|_{(n-k+1)} \leq u_j - 2 (\theta_j - \mu_j) \leq |u|_{(n-k+1)} \\
y_j \in  \mathcal{C}_j(k) \quad\Leftrightarrow\quad -|u|_{(n-k+1)} \leq u_j \leq |u|_{(n-k+1)} 
\end{gather*}
such that the result \eqref{mis:SRS:CSI} follows from applying Lemma \ref{lemma:SRS:pred}. 
\end{proof}

\noindent
Proof of Corollary \ref{coro:null}.
\begin{proof} For confidence inference of the marginal  $\theta = \mathbb{E}(y_i)$ with respect to $\mathcal{P}$, we notice that the training set mean $\bar{y}(s_1) = \tfrac{1}{n_1} \sum_{i\in s_1} y_i$ is an unbiased estimator of $\theta$ with respect to SRSWOR of $s_1$ from $s$. Given finite values of $y_i$ in the sample $s$, the SRSWOR variance of $\bar{y}(s_1)$ tends to $0$, asymptotically as $n_1 \rightarrow \infty$, i.e. $\bar{y}(s_1)$ is a consistent estimator of $\theta$. The result follows from the continuous mapping theorem.
\end{proof}


\begin{thebibliography}{999} 

\bibitem{barber2021} Barber, R.F., Candès, E.J., Ramdas, A. and Tibshirani, R.J. (2021). The limits of distribution-free conditional predictive inference. \textit{Information and Inference: A Journal of the IMA}, 10:455-482.

\bibitem{battese1988} Battese, G. E., Harter, R. M., and Fuller, W. A. (1988). An error component model for prediction of county crop areas using survey and satellite data. \textit{Journal of the American Statistical Association}, 83:28-36.

\bibitem{berger1998} Berger, Y. G. (1998). Rate of convergence for asymptotic variance of the Horvitz-Thompson estimator.
\textit{Journal of Statistical Planning and Inference}, 67:209-226.

\bibitem{bersson2024} Bersson, E, and Hoff, P.D. (2024). Optimal Conformal Prediction for Small Areas. \textit{Journal of Survey Statistics and Methodology}, 12:1464-1488,

\bibitem{fay1979} Fay, R.E. and Herriot, R.A. (1979). Estimates of income for small places: An application of James-Stein procedures to Census data. \textit{Journal of the American Statistical Association}, 74:269-277.

\bibitem{horvitz1952survey} Horvitz, D. G. and Thompson, D. J. (1952). A generalization of sampling without replacement from a finite universe. \emph{Journal of the American Statistical Association}, 47:663-685.
 
\bibitem{lei2018} Lei, J., G'Sell, M., Rinaldo, A., Tibshirani, R., and Wasserman, L. (2018). Distribution-Free Predictive Inference for Regression. \textit{Journal of the American Statistical Association}, 113:1094-1111. 

\bibitem{rao2015} Rao, J. N. and Molina, I. (2015). \emph{Small area estimation}. John Wiley \& Sons.
 
\bibitem{shafer2008} Shafer G. and Vovk, V. (2008). A Tutorial on Conformal Prediction. \textit{Journal of Machine Learning Research}, 9:371-421. 
 
\bibitem{tobshirani2019} Tibshirani, R.J., Barber, R.F., Candes, E. and Ramdas, A. (2019). Conformal prediction under covariate shift. \textit{Advances in Neural Information Processing Systems}, 32.

\bibitem{vovk2005-2} Vovk, V., Gammerman, A., and Shafer, G. (2022). \textit{Algorithmic learning in a random world, 2nd Edition}. Springer.

\bibitem{vovk2005} Vovk, V., Gammerman, A., and Shafer, G. (2005). \textit{Algorithmic Learning in a Random World}. Springer.


 
\end{thebibliography}
\end{document}